\documentclass[pdflatex,sn-mathphys-num]{sn-jnl}

\usepackage{graphicx}%
\usepackage{multirow}%
\usepackage{amsmath,amssymb,amsfonts}%
\usepackage{amsthm}%
\usepackage{mathrsfs}%
\usepackage[title]{appendix}%
\usepackage{xcolor}%
\usepackage{textcomp}%
\usepackage{manyfoot}%
\usepackage{booktabs}%
\usepackage{algorithm}%
\usepackage{algorithmicx}%
\usepackage{algpseudocode}%
\usepackage{listings}%

\theoremstyle{thmstyleone}%
\newtheorem{theorem}{Theorem}
\newtheorem{proposition}[theorem]{Proposition}%
\newtheorem{conjecture}[theorem]{Conjecture}%
\newtheorem{lemma}[theorem]{Lemma}%
\newtheorem{corollary}[theorem]{Corollary}%

\theoremstyle{thmstyletwo}%
\newtheorem{remark}{Remark}%

\theoremstyle{thmstylethree}%
\newtheorem{definition}{Definition}%

\def\F{\mathbb{F}}%
\def\hbreak{\vspace*{5pt}\hfil\break\noindent}%

\begin{document}

\title[Article Title]{Integrable cellular automata on finite fields of order $2^n$}


\author[1]{\fnm{Aoi} \sur{Araoka}}

\author*[1]{\fnm{Tetsuji} \sur{Tokihiro}}\email{t-toki@musashino-u.ac.jp}


\affil*[1]{\orgdiv{Department of Mathematical Engineering}, \orgname{Musashino University}, \orgaddress{\street{3-3-3 Ariake}, \city{Koto-ku}, \postcode{135-8181}, \state{Tokyo}, \country{Japan}}}




\abstract{This paper explores cellular automata (CA) constructed from Yang–Baxter maps over finite fields $\mathbb{F}_{2^n}$. We define $R$-matrices using a map $f$ on $\mathbb{F}_{2^n}$ and establish necessary and sufficient conditions for $f$ to satisfy the Yang–Baxter equation. We show that these conditions become remarkably streamlined in characteristic two. An exhaustive search for bijective solutions in fields of order 4, 8, and 16 yields 16, 736, and 269,056 maps, respectively.

Analysis of the resulting CA under helical boundary conditions reveals a consistent alignment between the temporal period and the field order. We propose the conjecture that this periodic identity holds generally for $\mathbb{F}_{2^n}$, supported by analytical proofs for $n=2$ and $n=3$. Our results further indicate that bijectivity is a fundamental requirement for this periodic behavior.}

\keywords{Cellular Automaton, Integrable system, finite field, characteristic two}



\maketitle

\section{Introduction}\label{sec1}

Cellular automata (CA) are discrete dynamical systems characterized by variables that evolve over a finite set in both time and space\cite{Wolfram2002}. Within this field, integrable cellular automata—exemplified by the Box-Ball System (BBS)—represent a particularly significant class\cite{Takahashi_Satsuma_1990}. Analogous to continuous nonlinear integrable equations, these discrete systems possess a sufficient number of independent conserved quantities to permit the analytical treatment of initial value problems.

Historically, such CAs have been derived through two primary methodologies: the ultradiscretization\cite{Tokihiro_1996} of nonlinear integrable equations and the crystallization of solvable lattice models\cite{Hatayama_2001}. In the latter framework, locally defined interactions satisfy the Yang–Baxter equation (YBE), which serves as a sufficient condition for the integrability of the lattice model\cite{Baxter_1982}. Consequently, constructing mappings that satisfy the YBE provides a robust and natural framework for developing integrable cellular automata.

While extensive research has been conducted on solvable lattice models and Yang–Baxter maps over the field of complex numbers, studies over finite fields remain relatively scarce. Given that CA states are inherently discrete, investigating their dynamics over finite fields—particularly those of characteristic two ($\mathbb{F}_{2^n}$)—is a compelling and natural direction. However, the specific conditions under which mappings over finite fields satisfy the YBE, as well as the resulting structural properties of the CAs, are not yet fully understood.

In this study, we focus on the Yang–Baxter map as a set-theoretical solution to the YBE\cite{Drinfeld_1992, veselov2003,kakei2010}. Let $K$ be an arbitrary field, and consider the map $R: K \times K \to K \times K$, defined by $(x,y) \mapsto (h,g)$ as follows:
\begin{equation}
(h(x,y), g(x,y)) = \left(\frac{y(1+\mu xy)}{1+\kappa x y},\ \frac{x(1+\kappa xy)}{1+\mu x y} \right),
\end{equation}
where $\mu, \kappa \in K$ are spectral parameters.
Furthermore, the following map, obtained by ultradiscretizing the aforementioned Yang-Baxter map, also satisfies the YBE:
\begin{align*}
&(h(x,y),g(x,y)) \\
&= (y+U_\mu(x+y)-U_\kappa(x+y),\ x+U_\kappa(x+y)-U_\mu( x+y) ) \\
&U_\kappa(x) :=\max[0, x+\kappa]
\end{align*}
Here, for the sake of simplicity, we omit the spectral parameters. By setting $U_\mu(x+y)-U_\kappa(x+y)=f(x+y)$,  we obtain the functional form:
\begin{align}
h(x,y) &= y + f(x+y), \\
g(x,y) &= x - f(x+y).
\end{align}
 This formulation allows us to characterize the $R$-matrix through a single function $f$ over the finite field $\mathbb{F}_{2^n}$.

The primary objective of this research is to construct integrable cellular automata over $\mathbb{F}_{2^n}$ based on the YBE and to elucidate their dynamical characteristics. Our methodology involves the following contributions:
\begin{itemize}
    \item \textbf{Exhaustive Search:} We derive the necessary and sufficient conditions for the $R$-matrix to satisfy the YBE and perform a computational search to identify all valid bijective mappings $f$ over fields of order 4, 8, and 16.
    \item \textbf{CA Construction:} Using these $R$-matrices, we develop cellular automata with \textbf{helical boundary conditions} (representing a spiral lattice structure) and analyze their time evolution.
    \item \textbf{Discovery of Periodicity:} We report a striking numerical observation: when $f$ is bijective, the period of the CA consistently coincides with the order of the underlying finite field, regardless of system size or initial configurations.
    \item \textbf{Mathematical Proof:} We formalize this observation by proposing the conjecture that CAs constructed in this manner over $\mathbb{F}_{2^n}$ always possess a period equal to the field order, providing analytical proofs for the cases of $\mathbb{F}_2, \mathbb{F}_4$, and $\mathbb{F}_8$.
\end{itemize}

\subsection{Related Works}

The integration of finite fields into the study of discrete integrable systems has been addressed by several key studies, though a gap remains regarding the $\mathbb{F}_{2^n}$ framework:
\begin{enumerate}
    \item \textbf{Bruschi et al. (1992,1994)\cite{Bruschi1992,bruschi1994cellular}:} The authors establish a general procedure to derive hierarchies of nonlinear cellular automata from the compatibility conditions of linear operators, applying concepts from soliton theory (Lax pairs) to discrete systems. Computer experiments confirm that these CA exhibit rich dynamics, including "particle-like" coherent structures. These structures behave as solitons: they are localized in space and preserve their shape and velocity even after interactions (collisions). This behavior is observed not only in 1+1 dimensions but also in 2+1 and 3+1 dimensions.
    \item \textbf{Białecki et al.(2003,2005)\cite{BialeckiDoliwa2003,BialeckiDoliwa2005}:} Both papers demonstrate that the geometric interpretation of integrable systems is robust enough to function over finite fields. While the construction methods (linear compatibility conditions, Riemann-Roch theorem) remain analogous to the complex case, the finite field context introduces unique properties such as exact periodicity and new types of soliton configurations based on Galois theory.
    \item \textbf{Kanki et al. (2012)\cite{Kanki2012}:} The paper investigates discrete integrable equations defined over finite fields. To resolve the indeterminacy of the time evolution, the authors propose an algebraic prescription. The proposed method provides a consistent way to define the time evolution of discrete integrable systems over finite fields for generic initial conditions. It successfully avoids the computational deadlocks found in previous methods and offers a unified perspective linking finite field arithmetic with the geometric theory of singularity confinement.
    \item \textbf{Yura (2014) \cite{Yura2014}} This paper proposes a new soliton system operating over finite fields by applying the mechanism of the "Box-Ball System" (BBS). Unlike conventional systems, this model is formulated using polynomials, providing the distinct advantage of avoiding instabilities such as division by zero. Its most prominent feature is the fractal-like nested structure within the solutions, which allows for stable propagation despite their complex internal morphology. Numerical simulations demonstrate that even after collisions, multiple solitons continue to propagate while maintaining their original shapes. By elucidating a new relationship between algebraic structures and integrability, this study offers a unique perspective on the theory of integrable systems.

\end{enumerate}

Our work builds upon these foundations by systematically exploring the $\mathbb{F}_{2^n}$ case within the Yang–Baxter framework, revealing a previously unknown periodic structure fundamental to these systems.

%
%
\section{Solutions to the Yang-Baxter equation on finite fields of order \texorpdfstring{$2^n$}{2\textasciicircum n}}
\subsection{The Yang-Baxter equation}
Let $X$ be a non-empty set. We consider a set-theoretical map $R$, often referred to as an $R$-matrix:
\begin{equation}\label{$R$-matrix}
\begin{aligned}
R: X \times X &\to X \times X \\
(x, y) &\mapsto (h(x,y), g(x,y)).
\end{aligned}
\end{equation}
Let $R_{ij}$ denote the map acting on the $i$-th and $j$-th components of the Cartesian product $X^3 = X \times X \times X$. For instance, $R_{12}$ and $R_{23}$ are defined as follows:
\begin{equation}
\begin{aligned}
R_{12}(x, y, z) &= (h(x,y), g(x,y), z), \\
R_{23}(x, y, z) &= (x, h(y,z), g(y,z)).
\end{aligned}
\end{equation}
The map $R$ is said to satisfy the Yang-Baxter equation (or the braid relation) if the following operator identity holds on $X^3$:
\begin{equation}\label{Yang-Baxter-relation}
R_{12} R_{23} R_{12} = R_{23} R_{12} R_{23}.
\end{equation}

These relations can be interpreted graphically, as shown in Figures~\ref{fig:a} and \ref{fig:b}.

\begin{figure}[htbp]
\centering
\begin{minipage}[b]{0.40\columnwidth}
    \centering
    \includegraphics[width=0.6\columnwidth]{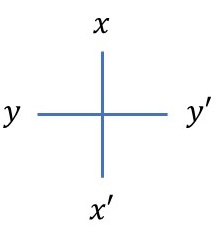}
    \caption{Graphical representation of the $R$-matrix.}
    \label{fig:a}
\end{minipage}
\hfill
\begin{minipage}[b]{0.5\columnwidth}
    \centering
    \includegraphics[width=0.9\columnwidth]{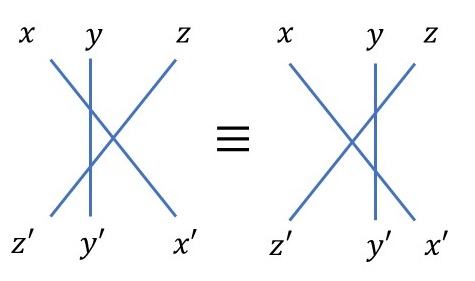}
    \caption{The Yang-Baxter equation represented as string interactions.}
    \label{fig:b}
\end{minipage}
\end{figure}

While the Yang-Baxter equation is conventionally studied over the field of complex numbers $\mathbb{C}$ in the context of quantum groups and solvable lattice models, this study explores its solutions and properties over finite fields $\mathbb{F}_q$.
%

\subsection{Necessary and sufficient conditions for the Yang–Baxter map}
Let us consider the necessary and sufficient conditions for $u, v, w$ in the Yang-Baxter equation.
Let the $R$-matrix be $(x,y) \to (x'=h(x,y), y'=g(x,y))$. Calculating the Yang-Baxter equation yields:
\begin{align}
h_{12}(h_{13}(x,g_{23}(y,z)),h_{23}(y,z))&=h_{13}(h_{12}(x,y),z)\\
g_{12}(h_{13}(x,g_{23}(y,z)),h_{23}(y,z))&=h_{23}(g_{12}(x,y),g_{13}(h_{12}(x,y),z))\\
g_{13}(x,g_{23}(y,z))&=g_{23}(g_{12}(x,y),g_{13}(h_{12}(x,y),g_{13}(h_{12}(x,y),z))
\end{align}
As an $R$-matrix satisfying this Yang-Baxter equation, we consider the map described by a single variable function $f$.
In the context of this study, we define the $R$-matrix as follows (corresponding to the ultradiscretized form with $f(x+y)$):
\begin{align}
R:\ (x,y) \, \mapsto \, (h(x,y),g(x,y))= \left(y+f(x+y),\ x-f( x+y) \right) \label{eq:finiteYB}
\end{align} 
The first task is to find the function $f$ that satisfies the Yang-Baxter equation.

\begin{proposition}
    \label{ybeq}
    The necessary and sufficient condition for $f$ to satisfy the Yang-Baxter equation is that for any $x, y \in \F_q$:
\begin{subequations}
\begin{align}
f(x-f(y))=f(x)+f(x-f(y+f(x)))\label{eq:YBf1}\\
f(y+f(x))=f(y)+f(y+f(x-f(y)))\label{eq:YBf2}
\end{align}
\end{subequations}
holds. In particular, when the order is $q=2^n$, the following equation holds:
\begin{equation}
f(x)+f\left(x+f(y)\right)=f\left(x+f(y+f(x))\right) \label{FYB_eq}
\end{equation}
\end{proposition}

\begin{proof}
    Equations \eqref{eq:YBf1} and \eqref{eq:YBf2} are obtained by explicitly calculating equations (4), (5), and (6). In the case of finite fields of characteristic two ($\mathbb{F}_{2^n}$), the additive and subtractive operations are identical ($x + y = x - y = x \oplus y$), and we have Eq.~\eqref{FYB_eq}.
\end{proof}

This simplification is a key finding of our study, as it significantly reduces the computational complexity required to identify valid mappings.
%
%

\subsection{Properties of the bijective mapping \texorpdfstring{$f$}{f}}

Let us consider the properties of the bijective mapping $f$ which satisfies \eqref{FYB_eq}.

\begin{definition}
    We define the sets $\Omega_a^{(b)}$ and $\Omega_a$ of maps $f: \F_{2^n} \rightarrow \F_{2^n}$ in $\F_{2^n}$ as follows:
\begin{align*}
\Omega_a^{(b)} &:=\left\{ f \, \big| \, f \text{ is a bijective map satisfying Eq. (\ref{FYB_eq}) and }\, f(a)=b \right\}\\
\Omega_a &:=\Omega_a^{(0)}
\end{align*}
\end{definition}
\begin{proposition}\label{prop1}
$$
{}^\forall a,\, {}^\forall b \in \F_2^n,\;\; \Omega_a^{(b)} \simeq \Omega_0^{(b)},\;\;\text{especially \ }\Omega_a \simeq \Omega_0
$$
\end{proposition}

\begin{proof}
    For any $f_0 \in \Omega_0^{(b)}$, define $f_a$ as ${}^\forall x \in \F_{2^n},\ f_a(x)=f_0(x+a)$. Then $f_a(x)$ satisfies \eqref{FYB_eq}.
    Since $f_a(a)=f_0(a+a)=f_0(0)=b$, we have $f_a \in \Omega_a^{(b)}$.
    Also, if two maps $f_0, f_0' \in \Omega_0^{(b)}$ are different, then $f_a(x):=f_0(x+a)$ and $f_a'(x):=f_0'(x+a)$ are different. Thus, the elements of $\Omega_0^{(b)}$ correspond to distinct elements of $\Omega_a^{(b)}$.

    Conversely, for any $f_a \in \Omega_a^{(b)}$, let $f_0(x):=f_a(x+a)$. Then $f_0(x)$ satisfies \eqref{FYB_eq}, and $f_0(0)=f_a(a)=b$, so $f_0 \in \Omega_0^{(b)}$. Similar correspondence exists.

    Thus, the elements of $\Omega_0^{(b)}$ and $\Omega_a^{(b)}$ are in $1:1$ correspondence. 
\end{proof}

\begin{proposition}\label{prop2}
If $f \in \Omega_0$, then ${}^\forall x\in \F_{2^n},\ f(f(x))=x$.
\end{proposition}

\begin{proof}
    In equation \eqref{FYB_eq}, setting $y=0$ gives:
\[
f(x)+f(x+0)=f(x+f(0+f(x)))\, \rightarrow \, 2f(x)=0=f(x+f(f(x)))
\]
Since the characteristic is 2, $2f(x)=0$. Due to the bijectivity of $f$ and $f(0)=0$, we have $x+f(f(x))=0$. Hence $f(f(x))=x$.
\end{proof}

\begin{proposition}\label{prop4}
If $f$ is a single transposition $[wz]$ $(w,z \in \F_{2^n}^\times)$, then $f\in\Omega_0$.
\end{proposition}

\begin{proof}
We demonstrate this by dividing cases based on whether $x, y$ coincide with $w, z$ in \eqref{FYB_eq}.
\begin{itemize}
\item If $x$ or $y$ is $0$, it holds automatically. Below, assume $x, y$ are not $0$.
\item Case where $x, y$ do not match $w, z$.
\begin{align*}
\mathrm{LHS \ of \ \eqref{FYB_eq}} &= x+f(x+y) \\
\mathrm{RHS \ of \ \eqref{FYB_eq}} &= f(x+f(x+y))
\end{align*}
If $x+f(x+y)=z$, then $f(x+y)=x+z$. Since $y \ne z$, we must have $x+y=z$ and $x+z=w$. This implies $y=x+x+y=x+z=w$, a contradiction.
Similarly, if $x+f(x+y)=w$, then $y=z$, a contradiction. Thus, $f(x+f(x+y))=x+f(x+y)$, and \eqref{FYB_eq} holds.
\item Case where $x \ne z, x \ne w$ and ($y=z$ or $y=w$).

Since both sub-cases are symmetric, let $y=z$.  
\begin{align*}
\mathrm{LHS} &= x+f(x+w) \\
\mathrm{RHS} &= f(x+f(x+z))
\end{align*}
If $x+w=z$, then $x+z=w$. Thus $x+f(x+w)=x+w=z$, and $f(x+f(x+z))=f(x+z)=f(w)=z$. It holds.
Similarly, if $x+z=w$, then $x+w=z$, so it holds.

If $x+w \ne z$, then $x+z \ne w$. Also, since $x \ne 0$, $x+w \ne w$. Thus $x+f(x+w)=x+x+w=w$. $f(x+f(x+z))=f(x+x+z)=f(z)=w$. 

Therefore, \eqref{FYB_eq} holds.
\item Case where $y \ne z, y \ne w$ and ($x=z$ or $x=w$).

Let $x=z$.
\begin{align*}
\mathrm{LHS} &= w+f(z+y) \\
\mathrm{RHS} &= f(z+f(y+w))
\end{align*}
If $z+y=w$, then $y+w=z$, $y=z+w$. Hence $w+f(z+y)=w+f(w)=w+z=y$, and $f(z+f(y+w))=f(z+w)=f(y)=y$. It holds.

If $z+y \ne w$, then $y+w \ne z$ and $z+y \ne z$. Thus $w+f(z+y)=w+z+y$, and $f(z+f(y+w)) =f(z+y+w)$. If $z+y+w=z$, then $y=w$, a contradiction. Similarly $z+y+w \ne w$, so $f(z+y+w)=z+y+w$. It holds.

Therefore, \eqref{FYB_eq} holds.
\item Case where ($x=y=z$) or ($x=y=w$).

Let $x=y=z$.
\begin{align*}
\mathrm{LHS} &= w+f(z+w) \\
\mathrm{RHS} &= f(z+f(z+w))
\end{align*}
Since $z+w \ne w, z$, we have $w+f(z+w)=w+z+w=z$, and $f(z+f(z+w))=f(z+z+w)=f(w)=z$. It holds.
\item Case where ($x=z, y=w$) or ($x=w, y=z$).

Let $x=z, y=w$.
\begin{align*}
\mathrm{LHS} &= w+f(z+z)=w+f(0)=w \\
\mathrm{RHS} &= f(z+f(w+w))=f(z+0)=w
\end{align*}
It holds.
\end{itemize}
Thus, the proposition is proved.
\end{proof}

\begin{proposition}\label{prop5}
If $f$ satisfies $f(f(x))=x$ and $f(x+y)=f(x)+f(y)$ for ${}^\forall x, {}^\forall y \in \F_{2^n}$, then $f\in \Omega_0$.
\end{proposition}

\begin{proof}
    In \eqref{FYB_eq} $f(x)+f\left(x+f(y)\right)=f\left(x+f(y+f(x))\right)$, noting that $f(f(x))=x$:
\begin{align*}
    \mathrm{LHS} &= f(x)+f(x)+f(f(y))=2f(x)+y=y\\
    \mathrm{RHS} &= f(x)+f(f(y))+f(f(f(x)))=f(x)+y+f(x)=y+2f(x)=y
\end{align*}
Thus it holds.
\end{proof}

\begin{proposition}\label{prop6}
    Let $m \in \F_{2^n}$ and $k \in \left(\F_{2^n}^\times \setminus \{m\}\right)$. If $f$ is a bijection on $\F_{2^n}$ satisfying $f(0)=0$ and the conditions:
\[
f(m)=m,\ f(k)=k+m, \ f(k+m)=k,\ f(x)=x+k \,(x \ne 0,m,k,k+m) 
\]
then $f\in \Omega_0$.
\end{proposition}

\begin{proof}
    We prove this by case analysis.
\begin{itemize}
\item Since it automatically holds for $x=0$ or $y=0$, assume $x \ne 0, y \ne 0$.
\item When $x=m$:
\begin{align*}
\mathrm{LHS} &= m+f(m+f(y)) \\
\mathrm{RHS} &= f(m+f(y+m))
\end{align*}
\begin{itemize}
\item If $y=m$, both sides become $m$, so it holds.
\item If $y=k$, $m+f(m+f(y))=m+f(m+k+m)=m+k+m=k$, $f(m+f(k+m))=f(m+k)=k$. It holds.
\item If $y=k+m$, similarly both sides become $k+m$. It holds.
\item If $y\not\in\{k,m,k+m\}$, then $m+f(m+f(y))=m+f(m+y+k)$.
Since $m+y+k \not\in\{k,m,k+m\}$, $m+f(m+y+k)=m+m+y+k+k=y$.
Similarly, $f(m+f(y+m))=f(m+y+m+k)=f(y+k)=y+k+k=y$. It holds. 
\end{itemize}
\item When $x=k$:
\begin{align*}
\mathrm{LHS} &= k+m+f(k+f(y)) \\
\mathrm{RHS} &= f(k+f(y+k+m))
\end{align*}
\begin{itemize}
\item If $y=m$, LHS $=k+m+f(k+m)=k+m+k=m$, RHS $=f(k+f(k))=f(k+m+k)=f(m)=m$. It holds.
\item If $y=k$, LHS $=k+m+f(k+k+m)=k+m+m=k$. RHS $=f(k+f(k+k+m))=f(k+m)=k$. It holds.
\item If $y=k+m$, LHS $=k+m+f(k+k)=k+m+f(0)=k+m$. RHS $=f(k+f(k+m+k+m))=f(k+f(0))=k+m$. It holds.
\item If $y\not\in\{k,m,k+m\}$, LHS $=k+m+f(k+y+k)=k+m+y+k=y+m$. Since $y+k+m \not\in \{k,m,k+m\}$, RHS $=f(k+y+k+m+k)=f(y+m+k)=y+m$. It holds.
\end{itemize}
\item When $x=k+m$, it follows similarly from symmetry $k \leftrightarrow k+m$.
\item When $x \not\in\{k,m,k+m\}$, noting $x+m, x+k, x+k+m \not\in\{m,k,m+k\}$:
\begin{align*}
\mathrm{LHS} &= x+k+f(x+f(y)) \\
\mathrm{RHS} &= f(x+f(y+x+k))
\end{align*}
\begin{itemize}
\item If $y=m$, LHS $=x+k+f(x+m)=x+k+x+m+k=m$. RHS $=f(x+m+x+k+k)=f(m)=m$. It holds.
\item If $y=k$, LHS $=x+k+f(x+k+m)=x+k+x+k+m+k=m+k$. RHS $=f(x+f(k+x+k))=f(x+x+k)=f(k)=k+m$. It holds.
\item If $y=k+m$, it holds similarly.
\item If $y \not\in\{m,k,k+m\}$, LHS $=x+k+f(x+y+k)$.

If $f(x+y+k) \in \{0, m, k, k+m\}$, then $x+f(x+y+k) \not\in\{m,k,m+k\}$, so RHS $=x+k+f(x+y+k)$, which holds.
If $f(x+y+k)\not\in\{0,m,k,k+m\}$, then $x+y+k \not\in\{0,k,m,k+m\}$, so LHS $=x+k+x+y+k+k=y+k$. RHS $=f(x+y+x+k+k)=f(y)=y+k$. It holds.
\end{itemize}
Thus, it holds in all cases.
\end{itemize}
\end{proof}

\subsection{Exhaustive search for valid mappings}
Utilizing the Propositions \ref{prop1} and \ref{prop2}, we performed an exhaustive computational search (brute-force search) to identify all bijective functions $f: \mathbb{F}_{q} \to \mathbb{F}_{q}$ that satisfy the Yang–Baxter equation for $q = 2^n$. The number of such solutions increases rapidly with the order of the field, as shown in the following theorem.
\begin{theorem}
    Among the bijective solutions to \eqref{FYB_eq} in Proposition 1, there are \textbf{16} solutions for $q=2^2=4$, \textbf{736} for $q=2^3=8$, and \textbf{269056} for $q=2^4=16$.
\end{theorem}

As for the relation to the properties of $f$, we have the following proposition.
\begin{proposition}\label{prop_total}
For $q=2^2$, all the bijective solutions to  \eqref{FYB_eq} are obtained from Proposition~\ref{prop5}, and those for $q=2^3$ are obtained from Propositions~\ref{prop4}, \ref{prop5} and \ref{prop6}.
\end{proposition}
\begin{proof}
 In the case of $q = 2^2 = 4$, the solutions within $\Omega_0$ that satisfy Proposition~\ref{prop5} are determined by the selection of the basis $\{a_1, a_2\}$ for the subspace $U$. There are $\binom{3}{2} = 3$ ways to choose such a basis. Including the identity map, we obtain 4 distinct solutions in $\Omega_0$. According to Proposition~\ref{prop1}, this leads to a total of $4 \times 4 = 16$ solutions, which accounts for all possible cases.
 
 For the case of $q=2^3=8$, the solutions in $\Omega_0$ are categorized as follows:
 \begin{enumerate}
 \item The number of solutions satisfying Proposition~\ref{prop4} is $\binom{7}{2} = 21$.
 \item Regarding the solutions satisfying Proposition~\ref{prop5}, the number of ways to choose a basis $\{a_1, a_2\}$ for $U$ is $\binom{7}{2} = 21$. Since there are 3 distinct bases that generate the same subspace $U$, there exist $21 / 3 = 7$ unique subspaces. For each $U$, there are 3 possible choices for $u_0$ (assuming $u_0 \neq 0$). Thus, we have $7 \times 3 = 21$ solutions. Adding the case where $u_0 = 0$ (the identity map), we obtain 22 solutions in total.
 \item For the solutions satisfying Proposition~\ref{prop6}, there are $7 \times 6 = 42$ cases for $m \neq 0$, and 7 cases for $k$ when $m = 0$. This results in $42 + 7 = 49$ solutions.
 \end{enumerate}
 Consequently, there are $21 + 49 + 22 = 92$ solutions in $\Omega_0$. By Proposition~\ref{prop1}, this implies a total of $92 \times 8 = 736$ solutions. Since this value is consistent with the count specified in the Theorem, we conclude that the solution set is exhaustive.   
\end{proof}

%
%
\section{Integrable cellular automata and their periods}
%
%
\subsection{Construction of Cellular Automata from \texorpdfstring{$R$}{R}-matrices}

We construct a cellular automaton (CA) using the $R$-matrices obtained in the previous section. 
The state of the CA at time step $t$ ($t \in \mathbb{Z}_{\ge 0}$) is denoted by $\phi(t) \in \mathbb{F}_q^{\times N}$:
\[
\phi(t) = (x_1(t), x_2(t), \dots, x_N(t)).
\]
Using an auxiliary state $\psi(t) \in \mathbb{F}_q^{\times N}$ and a left boundary value $b(t)$, we define the time evolution $\phi(t) \mapsto \phi(t+1)$ via the $R$-matrix, as illustrated in Fig.~\ref{Fig:CA_time_evolution}.

\begin{figure}[hbt]
\centering
\includegraphics[width=100mm]{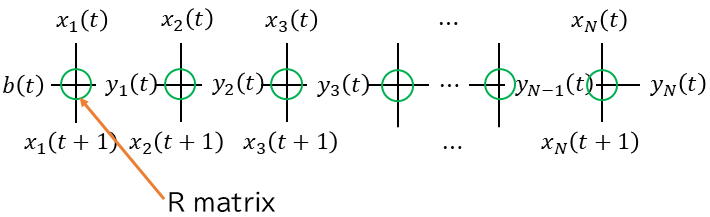}
\caption{The $R$-matrix defines the time evolution rule of the CA.}
\label{Fig:CA_time_evolution}
\end{figure}

Specifically, given the state $\phi(t)$ and the left boundary value $b(t)$, the $R$-matrix determines the updated state components $(x_i(t+1))$ and the auxiliary components $(y_i(t))$ for $i=1, 2, \dots, N$ according to the following mappings:
\[
R: (x_1(t), b(t)) \mapsto (x_1(t+1), y_1(t)), \quad (x_i(t), y_{i-1}(t)) \mapsto (x_{i}(t+1), y_i(t)) \quad (i=2,3,\dots,N).
\]
We impose a helical boundary condition, setting $b(t+1) = y_N(t)$ (see Fig.~\ref{Fig:CA_helical_BC}). 
Consequently, for a given initial state $\phi(0)$ and initial boundary value $b(0)$, we obtain a deterministic time evolution rule for the CA over $\mathbb{F}_q$.

\begin{figure}[hbt]
\centering
\includegraphics[width=0.8\linewidth]{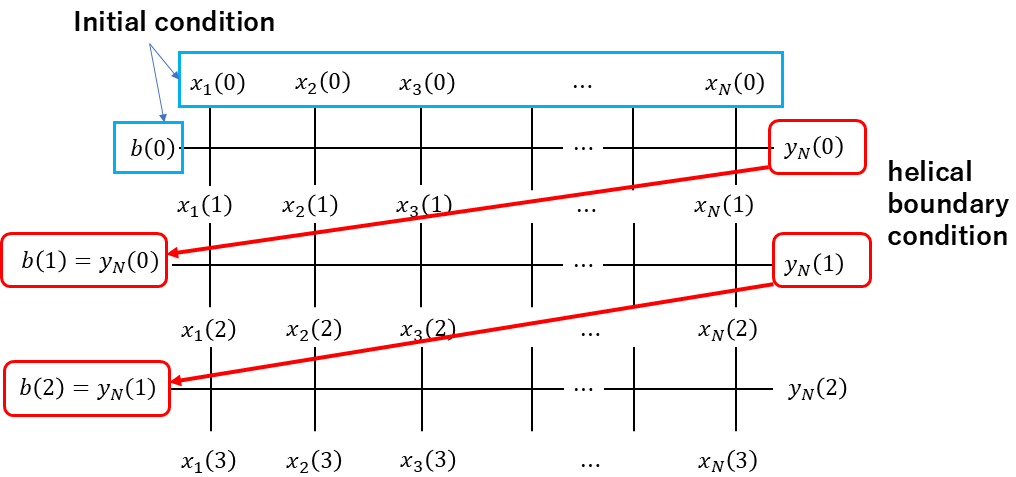}
\caption{Helical boundary condition and time evolution.}
\label{Fig:CA_helical_BC}
\end{figure}

Figure~\ref{Fig:example} demonstrates a concrete example over a finite field of order 4, with the initial state $[0,1,2,1,3,0,2,1]$, $b(0)=1$, and the bijection function 
\[
f = \begin{pmatrix}
0 & 1 & 2 & 3 \\
2 & 3 & 1 & 0
\end{pmatrix}.
\]
As depicted in the figure, we extract the sequence $\phi(t)$ to observe its behavior as a CA.

\begin{figure}[h]
\centering
\includegraphics[width=0.8\linewidth]{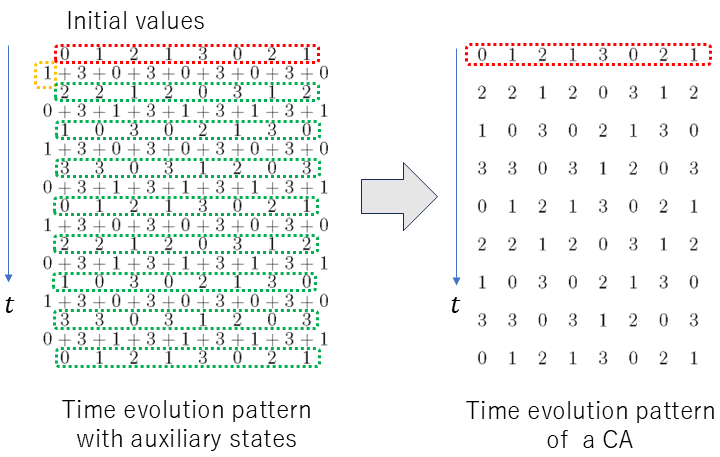}
\caption{Example of the CA time evolution over a finite field of order 4. The left panel displays the evolution pattern including auxiliary states, while the right panel shows the corresponding evolution pattern of the CA state alone.}
\label{Fig:example}
\end{figure}

In this instance, the initial state reappears at the 4th time step, indicating that the CA has a period of 4. Similar observations with other functions $f$ and initial values confirm that the possible periods are 1, 2, or 4.

Furthermore, for $q=2^3=8$ and $q=2^4=16$, all such CAs exhibit periods that are divisors of 8 and 16, respectively, regardless of the system size or initial state. Figure~\ref{fig:example2} provides examples for $q=8$, where each element is color-coded. In the right-hand panels, the initial values are displayed in white, and subsequent states are represented in grayscale based on their difference from the initial state directly above them, thereby making the periodicity visually apparent. An example for $q=16$ is shown in Fig.~\ref{fig:example3}.

\begin{figure}[htbp]
\centering
\includegraphics[scale=0.35]{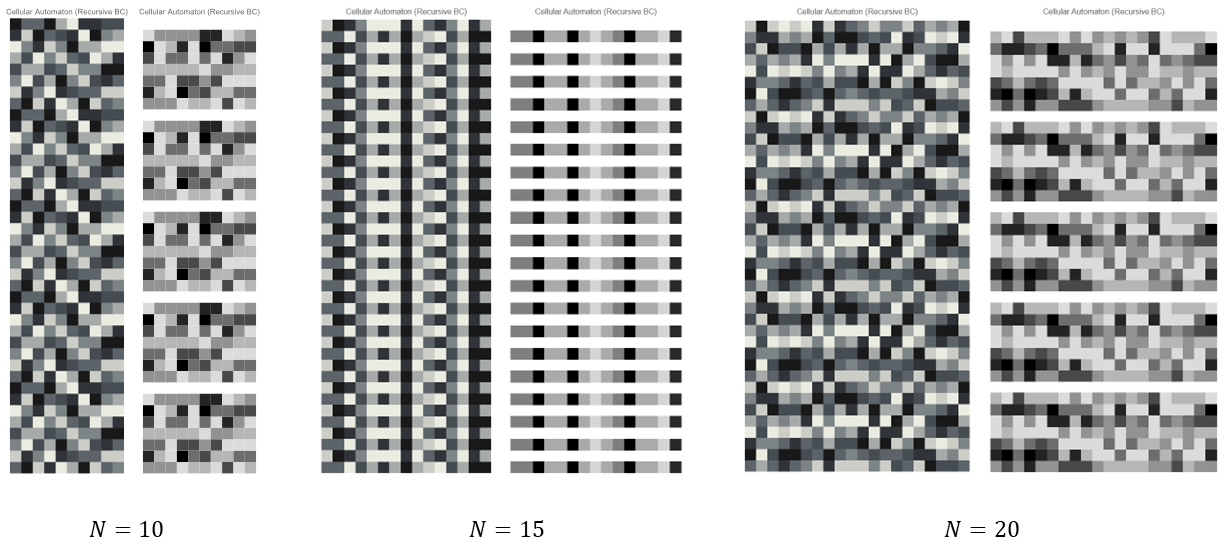}
\caption{Time evolution examples for $q=2^3=8$ with randomly chosen initial states. The right panels visualize the difference from the initial values. The CAs with sizes $N=10$ and $N=20$ have period 8, while the CA with $N=15$ has period 2.}
\label{fig:example2}
\end{figure}

\begin{figure}[htbp]
\centering
\includegraphics[scale=0.25]{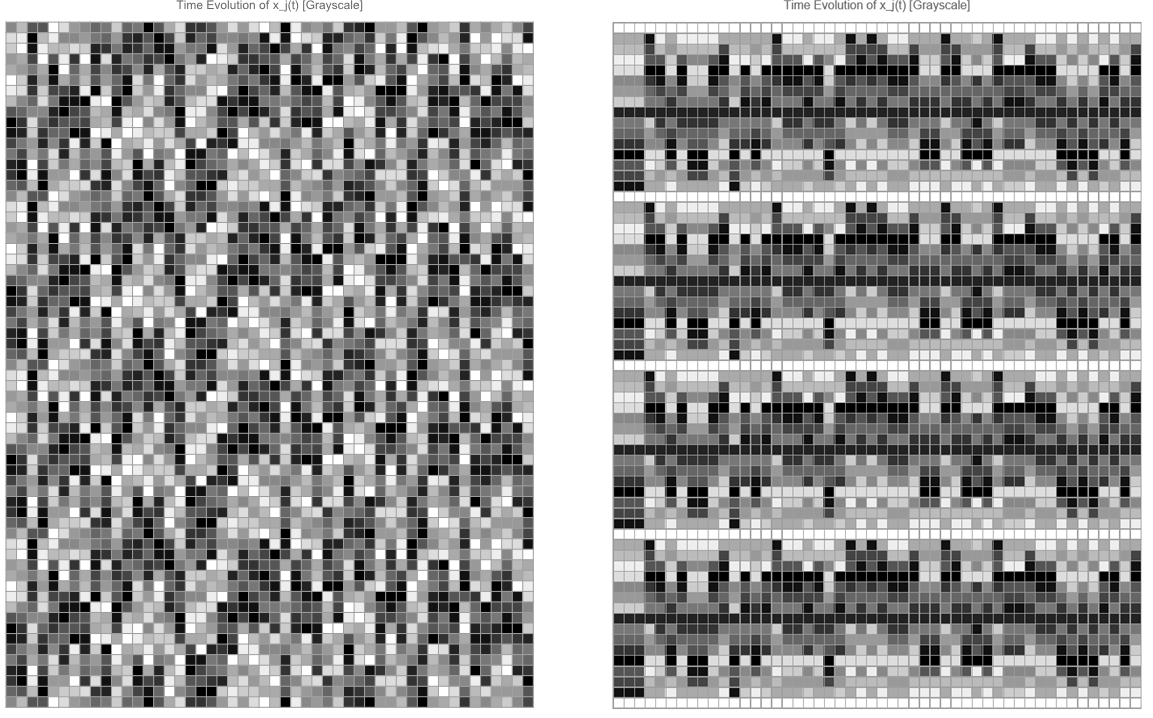}
\caption{Time evolution example for $q=2^4=16$ with a randomly chosen initial state. The right panel visualizes the difference from the initial values. The period is equal to $q=16$.}
\label{fig:example3}
\end{figure}

\begin{figure}[htbp]
\centering
\includegraphics[scale=0.25]{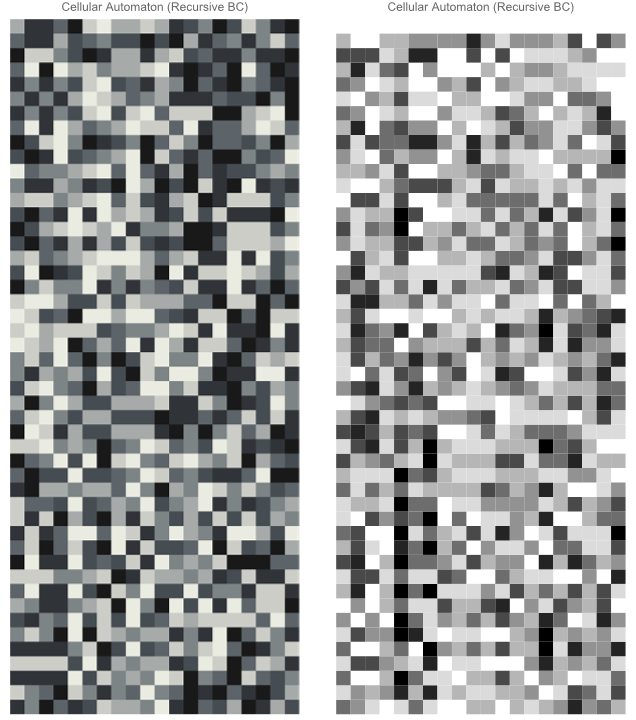}
\caption{Time evolution example of a non-integrable CA for $q=8$ with a randomly chosen initial state. The right panel visualizes the difference from the initial values. The period exceeds the simulated 48 time steps.}
\label{fig:example4}
\end{figure}

If we construct CAs using $R$-matrices that do not satisfy the Yang-Baxter equation, the resulting systems generally do not possess a period equal to $q$. Figure~\ref{fig:example4} illustrates such an example, where the $R$-matrix is defined by the bijective function
\[
f = \begin{pmatrix}
    0 & 1 & 2 & 3 & 4 & 5 & 6 & 7 \\
    0 & 3 & 5 & 6 & 1 & 2 & 7 & 4
\end{pmatrix},
\]
which violates Eq.~\eqref{FYB_eq}. We therefore conjecture that this strict periodicity is a distinctive feature of the integrable CAs.

Through numerical simulations, we determined the distribution of fundamental cycles for these CAs, which are summarized in Table~\ref{tab:distribution_period}. For example, there are 269,056 types of valid $R$-matrices for $q=16$, yielding CAs with fundamental cycles of 1, 2, 4, 8, or 16. These fundamental cycles depend on the initial conditions and the parity of the system size $N$, but are independent of the system size itself for $N \ge 2$\footnote{For $N=1$, there are 735 CAs with a fundamental cycle of 2 and none with a cycle of 4. Similarly, for $q=16$ (at $N=1$), there is only 1 CA with a fundamental cycle of 1, while the remaining 269,055 CAs have a cycle of 2.}. For classification purposes, we recorded the longest fundamental cycle observed among the various initial conditions.

\begin{table}[h]
\caption{The distribution of fundamental cycles ($N \ge 2$).}\label{tab:distribution_period}
\centering
\renewcommand{\arraystretch}{1.3}
\begin{tabular}{@{}c||c|c||c|c@{}}
\hline
Fundamental cycle & $q=8,\,$ $N$: odd & $q=8,\,$ $N$: even & $q=16,\,$ $N$: odd  & $q=16,\,$ $N$: even \\
\hline
1    & 1    & 1   & 1      & 1 \\
2    & 399  & 147 & 54015  & 11175 \\
4    & 336  & 420 & 174720 & 126840 \\
8    &      & 168 & 40320  & 110880 \\
16   &      &     &        & 20160 \\
\hline
\end{tabular}
\end{table}
%
%

\subsection{Main conjecture and proof for \texorpdfstring{$q=4,\,8$}{q=4, 8}}
From the observations in the previous subsection, we propose the following conjecture.

\begin{conjecture}
The cellular automaton thus constructed over a finite field of order $2^n$ has a period that is a divisor of the order of the field.
\end{conjecture}

We now prove this conjecture for $q=8$. The case $q=4$ follows directly from Corollary \ref{Cor:period2} and Proposition \ref{prop_line_p4}.

By Proposition 3, since $f(f(x))=x$, the elements of $\Omega^b_a$ can be represented as products of disjoint transpositions. Therefore, we proceed by case analysis based on the number of transpositions.

\begin{description}
    \item[1.]  \textbf{Case of 1 Transposition}:
    
    \begin{itemize}
        \item \textbf{Subcase} $\boldsymbol{f(0)=0}$: 
        Let the transposition be $[a\,b]$. We classify the elements as follows:
    \begin{align*}
    &U:=\{0,a,b,a+b\}, \quad  V=\{\gamma, a+\gamma, b+\gamma, a+b+\gamma\} \ (\gamma \not\in U)\\
    &U_0=\{0,a+b\}, \ U_1=\{a,b\}, \  V_0=\{\gamma,\gamma +a+b\}, \ V_1=\{\gamma+a,\gamma+b\} 
    \end{align*}
        \begin{lemma}\label{lem1-1}
            If $u, \, u' \in U$ and $v,\,v' \in V$, then:
            \begin{align*}
                &u+u' \in U,\quad v+v' \in U,\quad u+v \in V       
            \end{align*}
            Furthermore, let $u_i, u_i' \in U_i$ and $v_i,\,v_i' \in V_i$ for $i \in \{0,1\}$. If $i \ne j$, then:
            \begin{align*}
                &u_i+u_i' \in U_0,\quad u_i+u_j' \in U_1,\quad v_i+v_i' \in U_0,\\
                &v_i+v_j' \in U_1,\quad u_i+v_i \in V_0, \quad u_i +v_j \in V_1       
            \end{align*}
        \end{lemma}
        \hbreak
        \indent
            Let the action of the $R$-matrix be denoted by $R: x \otimes y \to y' \otimes x'$, represented diagrammatically as $\left[{\scriptstyle y}\mathop{+}\limits^{x}_{x'}{\scriptstyle y'}\right]$. Let $u,\,u' \in U$, $v,\,v' \in V$, $u_0,\,u_0' \in U_0$, etc. Then the following relations hold:
            \begin{align*}
                &u\otimes v \to v \otimes u \ \left[{\scriptstyle v}\mathop{+}\limits^{u}_{u}{\scriptstyle v}\right]  \ (\mathrm{identity}), \quad
                v \otimes u \to u \otimes v \ \left[{\scriptstyle u}\mathop{+}\limits^{v}_{v}{\scriptstyle u}\right] \ (\mathrm{identity}) \\
                &u_i\otimes u_i' \to u_i' \otimes u_i \ \left[{\scriptstyle u_i'}\mathop{+}\limits^{u_i}_{u_i}{\scriptstyle u_i'}\right]  \ (\mathrm{identity}), \quad
                v_i \otimes v_i' \to v_i' \otimes v_i \ \left[{\scriptstyle v_i'}\mathop{+}\limits^{v_i}_{v_i}{\scriptstyle v_i'}\right] \ (\mathrm{identity})\\ 
             &u_i\otimes u_j \to u_j^* \otimes u_i^* \ \left[{\scriptstyle u_j}\mathop{+}\limits^{u_i}_{u_i^*}{\scriptstyle u_j^*}\right]  \ (\mathrm{change}), \quad
                v_i \otimes v_j \to v_j^* \otimes v_i^* \ \left[{\scriptstyle v_j}\mathop{+}\limits^{v_i}_{v_i^*}{\scriptstyle v_j^*}\right] \ (\mathrm{change})
            \end{align*}
            where $i \ne j$ in the third line, and $x^*:=x+a+b$. That is, $u_0=0 \to u_0^*=a+b$, $v_1=\gamma+a \to v_1^*=\gamma+b$, etc. Summarizing these results leads to the following lemma.
            
         \begin{lemma}\label{lem1-2}
            For the map $R:\,x\otimes y \to y' \otimes x'$, if $(x, y) \in (U_i \times V_j) \cup (V_i \times U_j)$ with $i \ne j$, then $x'=x^*:=x+a+b$ and $y'=y^*:=y+a+b$. Otherwise, $x'=x$ and $y'=y$.
        \end{lemma}
        \begin{proposition}\label{prop1-1}
            The CA constructed from this $R$-matrix always has period 2.
        \end{proposition}
        \begin{proof}
            Suppose the boundary condition is $u_0\in U_0$. By Lemma \ref{lem1-2}, in the initial state $(x_1,x_2,...,x_N)$, only elements belonging to $U_1$ have the potential to change.
        All elements $u_1 \in U_1$ undergo a change. Due to the recursive boundary condition, in the second step, these same elements change again, returning the system to the initial state.
        Therefore, the period is 2.
        
        The argument for other initial boundary conditions is entirely analogous, yielding a period of 2.
        \end{proof}
        
\begin{proposition}\label{prop1-2}
    Consider a modified CA where the recursive boundary condition adds a constant $c \in \F_8$ at each recursion step. The period for this CA is always 4.
\end{proposition}
\begin{proof}
The subset (one of the four sets defined above) to which each column $x_i(t)$ belongs is invariant with respect to time $t$. Conversely, the subset to which the auxiliary state $y_i(t)$ belongs is either constant or oscillates with period 2. Consequently, by Lemma \ref{lem1-2}, the state of any column $x_i(t)$ follows a cycle of the form $x\to x \to x^* \to x^* \to x$, returning to the original value with period 4.
\end{proof}
 \item \textbf{Subcase} $\boldsymbol{f(0) \ne 0}$: \\
 There exist $f_0 \in \Omega_0$ and $c \in \F_8\setminus\{0\}$ such that $f(x)=f_0(x+c)$.\\[5pt]

 Direct calculation yields the following lemma.
        \begin{lemma}\label{lem1-3} Considering the effect of adding $c\in \F_8$:
        \begin{align*}
            c=a+b&:\,  x+c=x^* \\
            c\in U_1&:\,  x_i+c=x_j \ (i \ne j)\\
            c \in V_0&:\, u_i+c=v_i,\quad v_i+c=u_i \\
            c \in V_1&:\, u_i+c=v_j, \quad v_i+c=u_j \ (i \ne j)
        \end{align*} 
        Therefore, adding a constant induces a bijection on the sets $U_0,\,U_1,\,V_0,\,V_1$. 
        \end{lemma}
           \begin{lemma}\label{lem1-4}
               Let $R_0$ denote the action of the $R$-matrix corresponding to $f_0$.
               \begin{align*}
                   &R:\, x\otimes y \to y' \otimes x' \ \left[{\scriptstyle y}\mathop{+}\limits^{x}_{x'}{\scriptstyle y'}\right] \\
                 \Leftrightarrow \ &R_0:\  (x+c)\otimes y \to (y'+c) \otimes x' \;\; \left[{\scriptstyle y}\mathop{+}\limits^{(x+c)}_{x'}{\scriptstyle (y'+c)}\right] \\
              \Leftrightarrow \ &R_0:\  x \otimes (y+c) \to y' \otimes (x'+c) \;\; \left[{\scriptstyle (y+c)}\mathop{+}\limits^{x}_{(x'+c)}{\scriptstyle y'}\right] 
              \end{align*}
          \end{lemma} 
    \begin{proof}
    Let $R:\, x \otimes y \to y' \otimes x'$. Since $f(x+y)=f_0(x+y+c)$, we have $f(x+y)=x+y+a+b+c$ only when $x+y+c \in U_1$; otherwise, $f(x+y)=x+y+c$. Therefore,
    \begin{align*}
    x'=&
    \begin{cases}
       x+a+b+c & ( x+y+c \in U_1 )\\
       x+c & (\mathrm{otherwise})
    \end{cases}\\
    y'=&
    \begin{cases}
     y+a+b+c & (x+y+c \in U_1 )\\
     y+c & (\mathrm{otherwise})
    \end{cases}
    \end{align*}
    Thus, Lemma \ref{lem1-4} holds.
\end{proof}

\begin{definition}\label{def1-1}
Let the initial state be $(x_1(0),x_2(0),...,x_N(0))$ and the boundary condition be $y_1(0)$.
Let the state at the next time step be $(x_1(1),x_2(1),...,x_N(1))$ and the auxiliary state be $(y_2(0),y_3(0),...,y_{N+1}(0))$. By the recursive boundary condition, we set $y_{N+1}(0)=y_1(1)$.

Iterating this procedure, we define $(x_1(n+1),x_2(n+1),...,x_N(n+1))$, $(y_2(n),y_3(n),...,y_{N+1}(n))$, and $y_1(n+1)=y_{N+1}(n)$. 
\end{definition}            
\begin{proposition}\label{prp1-3}
The period is 2 when the system size $N$ is odd, and 4 when it is even.
\end{proposition}
\begin{proof}
Consider the transformation $x_i(n) \to x_i(n)+c$ where $i+n \equiv 1 \pmod 2$. Applying Lemma \ref{lem1-4}, the resulting CA pattern reduces to the case where $f(0)=0$. However, the effective boundary condition differs: when $N$ is odd, the boundary condition remains recursive, whereas when $N$ is even, the boundary condition involves adding $c$ at each recursion step. Therefore, by Propositions \ref{prop1-1} and \ref{prop1-2}, the periods are 2 and 4, respectively.
\end{proof}
    \end{itemize}
\item[2.] \textbf{Case of 2 Transpositions}:
\begin{itemize}
\item \textbf{Subcase} $\boldsymbol{f(0)=0}$:
Let $U:=\{x \in \F_8 \mid f(x)=x\}$ be the set of fixed points, and $V:=\F_8 \setminus U$. Note that $|U|=|V|=4$.
\begin{lemma}\label{lem2-1}
The set $U$ is closed under addition (i.e., $u_1+u_2 \in U$ for all $u_1,\,u_2 \in U$) if and only if $f$ is linear; that is, for any $x,\,y \in \F_8$:
\begin{equation}\label{eq2-0}
f(x+y)=f(x)+f(y).
\end{equation}
\end{lemma}
\begin{proof}
Suppose $u_1+u_2 \in U$ for all $u_1,\,u_2 \in U$. Since $0 \in U$, $U$ forms a subgroup of the additive group of $\F_8$. Thus, $U=\{0, u_1, u_2, u_1+u_2\}$. Consequently, $V$ is a coset of $U$, so there exists $\gamma \in \F_8 \setminus U$ such that $V=\{\gamma, \gamma+u_1, \gamma+u_2, \gamma+u_1+u_2\}$.

Since $f$ consists of two disjoint transpositions on $V$, without loss of generality, let the transpositions be $[\gamma, \gamma+u_1]$ and $[\gamma+u_2, \gamma+u_1+u_2]$. Under this assignment, for any $x \in V$, the action is given by $f(x)=x+u_1$.

We now verify linearity in all cases:
\begin{enumerate}
    \item If $x,\,y \in U$, then $x+y \in U$. Thus, $f(x+y)=x+y$ and $f(x)+f(y)=x+y$. The equality holds.
    \item If $x,\,y \in V$, then $x+y \in U$. Thus, $f(x+y)=x+y$. Meanwhile, $f(x)=x+u_1$ and $f(y)=y+u_1$, so $f(x)+f(y)=x+y+2u_1=x+y$ (since the field has characteristic 2). The equality holds.
    \item If $x \in U$ and $y \in V$, then $x+y \in V$. Thus, $f(x+y)=x+y+u_1$. Meanwhile, $f(x)=x$ and $f(y)=y+u_1$, so $f(x)+f(y)=x+y+u_1$. The equality holds.
\end{enumerate}
Thus, \eqref{eq2-0} holds.

Conversely, if \eqref{eq2-0} holds, then for any $u_1,\,u_2\in U$, we have $f(u_1+u_2)=f(u_1)+f(u_2)=u_1+u_2$. Therefore, $u_1+u_2$ is a fixed point, implying $u_1+u_2 \in U$.
\end{proof}

\begin{proposition}\label{prop2-1}
The function $f$ satisfies \eqref{eq2-0} for any $x,\,y \in \F_8$.
\end{proposition}
\begin{proof}
Let $U=\{0,u_1,u_2,u_3\}$. For the sake of contradiction, suppose $U$ is not closed under addition; specifically, suppose $u_3 \ne u_1+u_2$. Then $u_4 := u_1+u_2 \notin U$, so $u_4 \in V$. Let $f(u_4)=u_5$. Since $u_4 \in V$ and $f$ has no fixed points in $V$, $u_5 \ne u_4$.
Using the previously established functional equation \eqref{FYB_eq}:
\begin{align*}
    f(u_1)+f(u_1+f(u_2)) &= u_1+f(u_1+u_2) \\
    &= u_1+u_5, \\
    f(u_1+f(u_2+f(u_1))) &= f(u_1+f(u_2+u_1)) \\
    &= f(u_1+u_5).
\end{align*}
By \eqref{FYB_eq}, the left-hand sides correspond, implying $u_1+u_5$ must be a fixed point of $f$, so $u_1+u_5 \in U$.

The element $u_1+u_5$ must be either $u_2$ or $u_3$ (it cannot be $0$ or $u_1$ as that would imply $u_5=u_1$ or $u_5=0$, contradicting $u_5 \in V$).
\begin{itemize}
    \item If $u_1+u_5=u_2$, then $u_5=u_1+u_2$. However, $u_1+u_2=u_4$, so this implies $u_5=u_4$, which is a contradiction.
    \item Therefore, we must have $u_1+u_5=u_3$.
\end{itemize}
Similarly, swapping the roles of $u_1$ and $u_2$ leads to $u_2+u_5 \in U$. By the same logic, $u_2+u_5$ cannot be $u_1$, so $u_2+u_5=u_3$.

This implies $u_1+u_5 = u_2+u_5$, which yields $u_1=u_2$. This contradicts the distinctness of elements in $U$.

Therefore, the assumption was false, and $u_3=u_1+u_2$. Thus, $U$ is closed under addition. By Lemma \ref{lem2-1}, the proposition holds.
\end{proof}
\begin{proposition}
The Cellular Automaton (CA) constructed by a function $f$ satisfying \eqref{eq2-0} has a period of 2.
\end{proposition}
\begin{proof}
We use the notation from Definition \ref{def1-1}. Let us define the aggregate quantities $X$ and $Y$ as follows:
\begin{align*}
X &:= \sum_{k=1}^N \left( x_k(0)+f(x_k(0)) \right), \
Y &:= y_1(0)+f(y_1(0)).
\end{align*}
Since $f$ is linear (Prop. \ref{prop2-1}) and $f(f(z))=z$, we observe that $z+f(z)$ is always a fixed point of $f$. Thus, $f(X)=X$ and $f(Y)=Y$.
The state evolution is given by:
\begin{align*}
    x_k(1) &= f(x_k(0)) + Y, \\
    y_1(1) &= y_{N+1}(0) = X + 
    \begin{cases}
         y_1(0) & (N \ \text{even}) \\
         f(y_1(0)) & (N \ \text{odd}).
    \end{cases}
\end{align*}
We now calculate the state at time $t=2$:
\begin{align*}
    x_k(2) &= f(x_k(1)) + y_{N+1}(0) + f(y_{N+1}(0)) \\
    &= f\big(f(x_k(0))+Y\big) + X + f(X) + Y \\
    &= f(f(x_k(0))) + f(Y) + X + X + Y \\
    &= x_k(0) + Y + 0 + Y \\
    &= x_k(0).
\end{align*}
Similarly, for the boundary auxiliary state:
\begin{align*}
    y_{N+1}(1) &= \sum_{k=1}^N \big( x_k(1)+f(x_k(1)) \big) + 
    \begin{cases}
        y_1(1) & (N \ \text{even}) \\
        f(y_1(1)) & (N \ \text{odd}) 
    \end{cases} \\
    &= \sum_{k=1}^N \big( f(x_k(0)) + Y + f(f(x_k(0))+Y) \big) + (\text{Boundary Term}) \\
    &= \sum_{k=1}^N \big( f(x_k(0)) + Y + x_k(0) + f(Y) \big) + (\text{Boundary Term}) \\
    &= \sum_{k=1}^N \big( x_k(0) + f(x_k(0)) \big) + (\text{Boundary Term}) \quad (\text{since } Y+f(Y)=0) \\
    &= X + (\text{Boundary Term}).
\end{align*}
Substituting the boundary term:
\begin{itemize}
    \item If $N$ is even:
    \[ y_{N+1}(1) = X + (X + y_1(0)) = y_1(0). \]
    \item If $N$ is odd:
    \[ y_{N+1}(1) = X + f(X + f(y_1(0))) = X + f(X) + f(f(y_1(0))) = X + X + y_1(0) = y_1(0). \]
\end{itemize}
In both cases, the system returns to the initial state. Thus, the period is 2.
\end{proof}
The following corollary follows immediately.
\begin{corollary}\label{Cor:period2}
If $f(0)=0$ and $f$ consists of a product of 2 independent transpositions, the period of the CA is 2.
\end{corollary}

\item \textbf{Subcase} $\boldsymbol{f(0) \ne 0}$: \\
There exist $f_0 \in \Omega_0$ and a constant $c \in \F_8\setminus\{0\}$ such that $f(x)=f_0(x+c)$. Since $f_0$ is linear (by the previous subcase), we have $f(x) = f_0(x)+f_0(c)$. Let $\gamma = f_0(c)$. Then the relationship between the maps is given by:
\begin{align*}
   R:&\quad u \otimes v \rightarrow v' \otimes u'  
   \quad \left[{\scriptstyle v}\mathop{+}\limits^{u}_{u'}{\scriptstyle v'}\right]  \\
   \iff R_0:&\quad (u+\gamma) \otimes v \rightarrow (v'+\gamma) \otimes u' 
   \quad \left[{\scriptstyle v}\mathop{+}\limits^{u+\gamma}_{u'}{\scriptstyle v'+\gamma}\right] \\
   \iff R_0:&\quad u \otimes (v+\gamma) \rightarrow v' \otimes (u' +\gamma)
   \quad \left[{\scriptstyle v+\gamma}\mathop{+}\limits^{u}_{u'+\gamma}{\scriptstyle v'}\right]
\end{align*}
Therefore, analogous to the case of a single transposition, the following proposition holds.

\begin{proposition}\label{prop_line_p4}
The period of the CA is 2 when the system size $N$ is odd, and 4 when it is even.
\end{proposition}

\end{itemize}
\item[3.] \textbf{Case of 3 Transpositions}:

\begin{itemize}
\item \textbf{Subcase} $\boldsymbol{f(0)=0}$: 
Let $\F_8=\{u_j\}_{j=0}^7$ with $u_0=0$. Furthermore, let $u_1$ be the fixed point of $f$ (i.e., $f(u_1)=u_1$), and define $U=\{0,u_1\}$ and $V=\{u_j\}_{j=2}^7$.

\begin{remark}
    Recall Proposition \ref{prop6}, which was proved earlier. Note that this proposition encompasses the case where $m=0$. The condition for this specific case is that there exists some $k \ne 0$ such that:
    \[
    f(0)=0,\quad f(k)=k,\quad \text{and} \quad f(x)=x+k \quad \text{for all } x \not\in\{ 0,k\}.
    \]
\end{remark}

\begin{proposition}\label{prop3-1}
    A necessary and sufficient condition for $f$ to satisfy \eqref{FYB_eq} is that one of the 3 transpositions $[u_i, u_j]$ composing $f$ satisfies $u_i+u_j=u_1$.
\end{proposition}

\begin{proof}
    \textbf{Sufficiency:}
    Assume the transposition $[u_2, u_3]$ in $f$ satisfies $u_2+u_3=u_1$. Without loss of generality, let the other transpositions be $[u_4, u_5]$ and $[u_6, u_7]$.
    
    Let $\Omega_0=\{0, u_1, u_2, u_3\}$. Since $u_2+u_3=u_1$, $\Omega_0$ is a subgroup of the additive group of $\F_8$. Let $\Omega_1=\{u_4, u_5, u_6, u_7\}$ be the complement of $\Omega_0$. From the properties of vector spaces, $\Omega_1$ is the coset of $\Omega_0$. Thus, for any $v \in \Omega_1$, the relation $\Omega_1=\{u+v \mid u \in \Omega_0\}$ holds.
    
    Consequently, the relationship between $u_4$ and $u_5$ must fall into one of the following three cases based on the value of their sum (which must be in $\Omega_0 \setminus \{0\}$):
    \[
    (a)\ u_5=u_4+u_1,\quad (b)\ u_5=u_4+u_2,\quad (c)\ u_5=u_4+u_3.
    \]
    In case ($a$), we have $u_4+u_5=u_1$. Since $u_2+u_3=u_1$ by assumption, and the sum of all elements in $\F_8$ is 0, the remaining pair must also satisfy $u_6+u_7=u_1$. Thus, $f(x)=x+u_1$ for all $x \in V$, which corresponds to Proposition \ref{prop6} with $m=0$ and $k=u_1$.
    Similarly, case ($b$) corresponds to $m=u_1, k=u_2$, and case ($c$) corresponds to $m=u_1, k=u_3$. In all cases, $f$ satisfies \eqref{FYB_eq}.

    \textbf{Necessity:}
    Conversely, assume that no transposition $[u_i, u_j]$ satisfies $u_i+u_j=u_1$. Let the transpositions constituting $f$ be $[u_2, u_3]$, $[u_4, u_5]$, and $[u_6, u_7]$. Without loss of generality, assume $u_2+u_4=u_1$.
    
    Let $x=u_1$ and $y=u_2$. Evaluating the terms in \eqref{FYB_eq} yields:
    \begin{align*}
    f(u_1+f(u_2+f(u_1))) &= f(u_1+f(u_1+u_2)) \\
                         &= f(u_1+f(u_4)) = f(u_1+u_5), \\
    f(u_1)+f(u_1+f(u_2)) &= u_1+f(u_1+u_3).
    \end{align*}
    Let $\Omega_0=\{0, u_1, u_2, u_4\}$ and $\Omega_1=\{u_3, u_5, u_6, u_7\}$. Note that $\Omega_0$ is a subgroup. If $u, u' \in \Omega_0$ and $v, v' \in \Omega_1$, then $u+u' \in \Omega_0$, $v+v' \in \Omega_0$, and $u+v \in \Omega_1$.
    
    Thus, $u_1+u_5 \in \Omega_1 \setminus \{u_5\} = \{u_3, u_6, u_7\}$ and similarly $u_1+u_3 \in \{u_5, u_6, u_7\}$. Note that $u_1+u_5 \ne u_1+u_3$.
    
    We consider the possible values for $u_1+u_5$:
    \begin{itemize}
        \item If $u_1+u_5=u_3$, then $f(u_1+u_5)=f(u_3)=u_2$. In this case, $u_1+u_3=u_5$. Substituting this into the RHS expression: $u_1+f(u_1+u_3)=u_1+f(u_5)=u_1+u_4=u_2$. (Note: This case implies consistency locally, but conflicts with the assumption that no transposition sums to $u_1$ in the global constraints of $\F_8$).
        \item If $u_1+u_5=u_6$, then $f(u_1+u_5)=f(u_6)=u_7$. In this case, $u_1+u_3=u_7$ (since the sum of elements must balance). The RHS becomes $u_1+f(u_1+u_3)=u_1+f(u_7)=u_1+u_6=u_5$. Since $u_7 \ne u_5$, \eqref{FYB_eq} does not hold.
        \item The case $u_1+u_5=u_7$ follows similarly and leads to a contradiction.
    \end{itemize}
    Thus, it is shown that if no transposition $[u_i, u_j]$ satisfies $u_i+u_j=u_1$, \eqref{FYB_eq} does not hold.
\end{proof}

From the above, the following corollary holds for any function $f$ composed of 3 transpositions that satisfies \eqref{FYB_eq}.

\begin{corollary}\label{Cor3-1}
    The case of 3 transpositions reduces to the conditions described in Proposition \ref{prop6}.
\end{corollary}

By Corollary \ref{Cor3-1}, for the case of 3 transpositions, we only need to handle the following two cases:
\begin{enumerate}
    \item $f(0)=0,\,f(m)=m,\, f(x)=x+m,\, (x \ne \{0,m\})$
    \item $f(0)=0,\,f(m)=m,\, f(k)=k+m,\,f(k+m)=k,\,f(x)=x+k,\, (x \ne \{0,m,k,m+k\})$
\end{enumerate}

Let $U=\{0,m,k,m+k\}$, $U_0=\{0,m\}$, $U_1:=\{k,k+m\}$, $V=\{\gamma,\gamma+m,\gamma+k,\gamma+m+k\}$, $V_0=\{\gamma, \gamma+m\}$, $V_1:=\{\gamma+k, \gamma+m+k\}$.
\\[10pt]
\textbf{Case 1}:\\
The action of the $R$-matrix is as follows. As before, $u \in U, v \in V, u_i, u_i' \in U_i, V_i, V_i' \in V_i$ ($i=0,1$).
\begin{align*}
&\left[{\scriptstyle u_0'}\mathop{+}\limits^{u_0}_{u_0}{\scriptstyle u_0'}\right], \quad   
\left[{\scriptstyle u_1'}\mathop{+}\limits^{u_1}_{u_1}{\scriptstyle u_1'}\right], \quad 
\left[{\scriptstyle u_0}\mathop{+}\limits^{u_1}_{u_1+m}{\scriptstyle u_0+m}\right], \quad 
\left[{\scriptstyle u_1}\mathop{+}\limits^{u_0}_{u_0+m}{\scriptstyle u_1+m}\right]\\
&\left[{\scriptstyle v_0'}\mathop{+}\limits^{v_0}_{v_0}{\scriptstyle v_0'}\right], \quad   
\left[{\scriptstyle v_1'}\mathop{+}\limits^{v_1}_{v_1}{\scriptstyle v_1'}\right], \quad 
\left[{\scriptstyle v_0}\mathop{+}\limits^{v_1}_{v_1+m}{\scriptstyle v_0+m}\right], \quad 
\left[{\scriptstyle v_1}\mathop{+}\limits^{v_0}_{v_0+m}{\scriptstyle v_1+m}\right]\\
&\left[{\scriptstyle u}\mathop{+}\limits^{v}_{v+m}{\scriptstyle u+m}\right], \quad 
\left[{\scriptstyle v}\mathop{+}\limits^{u}_{u+m}{\scriptstyle v+m}\right]  
\end{align*}
\\[5pt]
\textbf{Case 2}:\\
The action of the $R$-matrix is as follows.
\begin{align*}
&\left[{\scriptstyle u_0'}\mathop{+}\limits^{u_0}_{u_0}{\scriptstyle u_0'}\right], \quad   
\left[{\scriptstyle u_1'}\mathop{+}\limits^{u_1}_{u_1}{\scriptstyle u_1'}\right], \quad 
\left[{\scriptstyle u_0}\mathop{+}\limits^{u_1}_{u_1+m}{\scriptstyle u_0+m}\right], \quad 
\left[{\scriptstyle u_1}\mathop{+}\limits^{u_0}_{u_0+m}{\scriptstyle u_1+m}\right]\\
&\left[{\scriptstyle v_0'}\mathop{+}\limits^{v_0}_{v_0}{\scriptstyle v_0'}\right], \quad   
\left[{\scriptstyle v_1'}\mathop{+}\limits^{v_1}_{v_1}{\scriptstyle v_1'}\right], \quad 
\left[{\scriptstyle v_0}\mathop{+}\limits^{v_1}_{v_1+m}{\scriptstyle v_0+m}\right], \quad 
\left[{\scriptstyle v_1}\mathop{+}\limits^{v_0}_{v_0+m}{\scriptstyle v_1+m}\right]\\
&\left[{\scriptstyle u}\mathop{+}\limits^{v}_{v+k}{\scriptstyle u+k}\right], \quad 
\left[{\scriptstyle v}\mathop{+}\limits^{u}_{u+k}{\scriptstyle v+k}\right]  
\end{align*}
Note that $\boldsymbol{+m}$ induces element swapping within each small group ($\{U_i,\,V_i\}$ ($i=0,1$), and $+k$ induces swapping between small groups $U_0 \leftrightarrow U_1$, $V_0 \leftrightarrow V_1$.

\begin{proposition}\label{prop3-2}
  In Case 1, the CA period is 2. In Case 2, the CA period is 4.  
\end{proposition}
\begin{proof}
In Case 1, if the small group to which the boundary condition $y_1(0)$ belongs is determined, the auxiliary field $y_k(t)$ simply permutes within that small group. Therefore, the initial state $(x_k(0))$ $(k=1,2,...,N)$ must return to the original state after 2 time steps.

In Case 2, the small group to which the auxiliary field $y_k(t)$ at column $k$ belongs returns to the original group with at most period 2, and whether it belongs to group $U$ or $V$ does not change. Therefore, $x_k(t)$ also returns to the original state with at most period 4.
\end{proof}
Regarding the boundary conditions (\{$y_1(t)$\}), the following proposition is important.
\\
\begin{proposition}\label{prop3-ex}
   Assume as a boundary condition that the small group to which $y_1(t)$ ($t=0,1,2,...$) belongs changes with period 4, and the group ($U$ or $V$) is either constant or changes with period 2. The period of the CA determined by this is 8.
\end{proposition}

\begin{remark}
    The condition of the proposition implies that the small group of $y_1(t)$ changes in a 4-cycle such as $U_0, V_1, U_1, V_1$ or $U_0, U_0, U_1, U_0$.
\end{remark}

\begin{proof}
Note that the group ($U$ or $V$) to which $x_1(t)$ belongs does not change with time step $t$.
\\[5pt]
1). $x_1(t)$ and $x_1(t+4)$ belong to the same small group. Also, $x_1(t)=x_1(t+8)$.

 Reason: During the period of 4, it interacts with different groups via the $R$-matrix an even number of times. Therefore, the change between small groups occurs an even number of times, so $x_1(t)$ and $x_1(t+4)$ belong to the same small group. Also, since interactions with different small groups within the same group occur in the same order and same number of times, $x_1(t)=x_1(t+8)$.\\[5pt]

2). The small group to which $y_i(t)$ ($i=2,3,...$) belongs also changes with period 4, and the group is constant or changes with period 2.
Therefore, $x_i(t)$ ($i=2,3,...$) has period 8.

 Reason: Since the group of $y_i(t)$ is the same as $y_1(t)$, the group is constant or changes with period 2.
Also, since the small group of both $x_1(t)$ and $y_1(t)$ changes with period 4, $y_2(t)$ also changes its small group with period 4. $x_2(t)$ and $x_2(t+4)$ belong to the same small group. Also, $x_2(t)=x_2(t+8)$.
Repeating this, for any $i$, $x_i(t+8)=x_i(t)$, so the period of this CA is 8.
\end{proof}
%


Next, we consider the system under the recursive boundary condition $y_1(t+1)=y_{N+1}(t)+c$ for a given constant $c$.

\begin{proposition}\label{prop3-3}
    Under this boundary condition, the period of the CA is 8.
\end{proposition}

\begin{proof}
    The proof proceeds in two steps:
    
    \begin{enumerate}
    \item \textbf{The large group ($U$ or $V$) to which $y_1(t)$ belongs is either constant or changes with period 2.}
    
    \textit{Reason:} The large group assignment for $y_i(t)$ is consistent across $i=1, \dots, N+1$ for a fixed $t$. Furthermore, adding a constant $c$ to any element in a group (e.g., $U$) consistently maps it either back to the same group (if $c \in U$) or to the complementary group (if $c \in V$). Consequently, the group assignment of $y_1(t+1)$ relative to $y_1(t)$ is determined solely by $c$. If $c \in U$, the group is constant; if $c \in V$, the group alternates, resulting in a period of 2. Thus, claim 1) holds.
    
    \item \textbf{The small group (subset) to which $y_1(t)$ belongs changes with period 4.}
    
    \textit{Reason:} We distinguish two cases based on the large groups:
    \begin{itemize}
        \item \textbf{Case A:} $y_{N+1}(0)+c$ belongs to the same large group as $y_1(0)$. 
        Recall that transitions (shifts) between small groups occur only via interaction with a different large group. Since the large group assignment of $x_i(t)$ is time-independent, the sequence of interactions is stable. Consequently, $y_i(0)$ and $y_i(1)$ experience the same number of small group transitions as $i$ varies from $1$ to $N+1$. Therefore, $y_1(2)$ returns to the same small group as $y_1(0)$, implying $y_1(t)$ has period 2.
        
        \item \textbf{Case B:} $y_{N+1}(0)+c$ belongs to a different large group than $y_1(0)$.
        By similar reasoning, while $y_i(0)$ and $y_i(1)$ may differ, $y_i(0)$ and $y_i(2)$ will experience the same parity of transitions relative to their starting positions. Similarly, $y_i(1)$ and $y_i(3)$ are paired. Therefore, $y_1(0)$ and $y_1(4)$ will belong to the same small group. In this case, $y_1(t)$ has period 4.
    \end{itemize}
    Combining these cases, the small group to which $y_1(t)$ belongs changes periodically with a period of 4 (as 2 divides 4).
    \end{enumerate}

    From steps 1) and 2), we conclude that the auxiliary field $y_1(t)$ changes its small group assignment with a period of 4. Therefore, applying Proposition \ref{prop3-ex}, the total period of the CA is 8.
\end{proof}

\item \textbf{Subcase} $\boldsymbol{f(0) \ne 0}$: \\
\begin{proposition}
    Under recursive boundary conditions, if the system size $N$ is odd, Case 1 yields a period of 2 and Case 2 yields a period of 4. If $N$ is even, the period is 8.
\end{proposition}

\begin{proof}
    The proof follows by analogous reasoning to the case of a single transposition (Proposition \ref{prp1-3}) and applying the results of Proposition \ref{prop3-3}.
\end{proof}

\end{itemize}
\end{description}

%
%
\section{Concluding Remarks}
In this paper, we considered CAs over finite fields of order $2^n$ constructed from the $R$-matrices satisfying the Yang-Baxter equation with a helical boundary condition.  By imposing the bijection condition for a function $f$ used to determine the $R$-matrices, we conjectured that the CAs have the same period of the order of the finite fields. We derived the total number of solutions for the function $f$ and showed effective theorems, and proved the conjecture for the cases of order 4 and 8.

 When we relax the condition of bijection, the number of functions $f$ satisfying the Yang-Baxter relation (\eqref{FYB_eq}) is $18664$. 
In the case of $f = \begin{pmatrix}
0 & 1 & 2 & 3 & 4 & 5 & 6 & 7\\
0 & 0 & 7 & 0 & 4 & 7 & 0 & 5
\end{pmatrix}$ (rule800):\\
The period is 5 for initial value (1,2,3,5), 46 for initial value (2,3,5,7), and 78 for initial value (3,4,6,1,7,6,7,1). The period becomes increasingly long as the system size increases.
Therefore, if $f$ is not bijective, the conjecture does not hold.
Also, while characteristic solutions such as soliton solutions could not be confirmed at present, it seems likely that solutions with some particle nature exist, as shown in Figure \ref{fig:nonbijectiveexample}.

\begin{figure}[htb]
\centering
\includegraphics[width=0.4\linewidth]{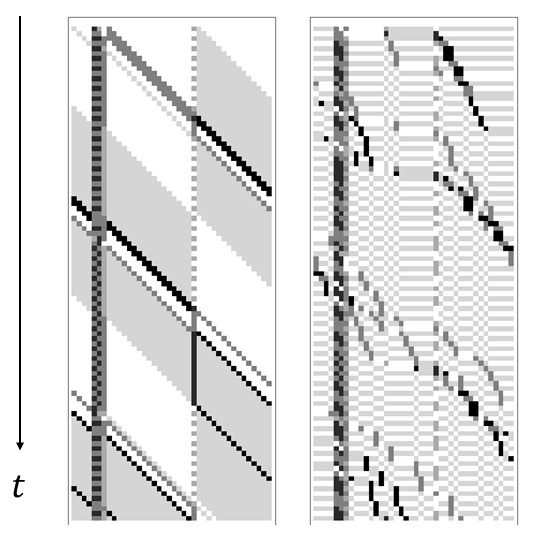}
\caption{Example of time evolution patterns using non-bijective $f$. No periodicity is observed.}
\label{fig:nonbijectiveexample}
\end{figure}

Future work includes verifying whether the conjecture holds for general $\F_{2^n}$. Additionally, we aim to construct the function $f$ similarly for general finite fields and to extend the target on which the $R$-matrix acts to vector spaces or projective spaces over finite fields.

\bibliography{ref_Araoka}

\end{document}